\documentclass[12pt,notitlepage]{iopart}
\usepackage{iopams}
\usepackage{ytableau}
\usepackage{bbm}
\usepackage{amsthm}
\usepackage[utf8]{inputenc}
\usepackage{tensor}
\usepackage{multirow}
\usepackage{slashbox}
\newcommand{\node}{\mathfrak{n}}
\newcommand{\link}{\mathfrak{l}}
\newcommand{\Hil}{\mathcal{H}}

\newcommand{\nodes}[1]{{\rm Nodes}(#1)}
\newcommand{\links}[1]{{\rm Links}(#1)}
\newcommand{\iu}{\rmi}
\newcommand{\Volume}{\hat{V}}
\newcommand{\sgn}{{\rm sgn}}
\newcommand{\C}[1]{\mathbb{C}#1}
\newcommand{\CG}{\mathbb{C}G}
\newcommand{\SU}{{\rm SU(}2{\rm)}}
\newcommand{\GL}[1]{{\rm GL(}#1{\rm )}}
\newcommand{\End}[1]{{\rm End(}#1{\rm )}}
\newcommand{\id}{\mathbbm{1}}
\newcommand{\braket}[2]{\left< #1 | #2 \right>}
\newcommand{\tjsymbol}[6]{
\left(\begin{array}{ccc}
   #1 & #2 & #3 \\
   #4 & #5 & #6 
\end{array}\right)}
\newcommand{\sjsymbol}[6]{\left\{\begin{array}{ccc}
   #1 & #2 & #3 \\
   #4 & #5 & #6 
\end{array}\right\}}
\newcommand{\binom}[2]{{{#1}\choose{#2}}}
\newcommand{\ket}[1]{| #1 >}

\def\be{\begin{equation}}
\def\ee{\end{equation}}
\newcommand{\Inv}[1]{{\rm Inv}\left( #1 \right)}
\newcommand{\InvL}[2]{{\rm Inv}_{#2}\left( #1 \right)}
\newcommand{\InvLambda}[1]{\InvL{#1}{T_\lambda}}
\newtheorem{theorem}{Theorem}
\newtheorem{lemma}{Lemma}
\def\eqref{\eref}
\begin{document}
\title[Rovelli-Smolin-DePietri volume operator for monochromatic intertwiners]{Properties of the Rovelli-Smolin-DePietri volume operator in the spaces of monochromatic intertwiners}
\author{Marcin Kisielowski}
\address{National Centre for Nuclear Research, Pateura 7, 02-093 Warsaw, Poland}
\begin{abstract}
We study some properties of the Rovelli-Smolin-DePietri volume operator in loop quantum gravity, which significantly simplify the diagonalization problem and shed some light on the pattern of degeneracy of the eigenstates. The operator is defined by its action in the spaces of tensor products $\Hil_{j_1}\otimes \ldots \otimes \Hil_{j_N}$ of the irreducible SU(2) representation spaces $\Hil_{j_i}, i=1,\ldots,N$, labelled with spins $j_i\in \frac{1}{2}\mathbb{N}$. We restrict to spaces of SU(2) invariant tensors (intertwiners) with all spins equal $j_1=\ldots=j_N=j$. We call them spin $j$ monochromatic intertwiners. Such spaces are important in the study of SU(2) gauge invariant states that are isotropic and can be applied to extract the cosmological sector of the theory. In the case of spin $1/2$ we solve the eigenvalue problem completely: we show that the volume operator is proportional to identity and calculate the proportionality factor. 
\end{abstract}
\maketitle
\section{Introduction}
The idea that isotropic states in loop quantum gravity should be described by monochromatic intertwiners appeared in the literature a couple of times. In spin-foam cosmology the conditions of homogeneity and isotropy are imposed on the coherent states \cite{BRV}. The coherent states are linear combinations of spin-network states for different spins but the main contribution comes from states with all spins equal. As a result, the leading order amplitudes involve the monochromatic intertwiners only. In group field theory the states are restricted to isotropic states by considering monochromatic intertwiners maximizing the volume \cite{CosmoGFTI,CosmoGFTII}. In our recent approach \cite{HomogeneousIsotropicLQG} we study homogeneous-isotropic invariant subspaces of the scalar constraint operators proposed in \cite{LewandowskiSahlmann,HamiltonianOperator,NewScalarConstraint,TimeEvolution}. The component of the homogeneous-isotropic invariant space corresponding to spin network with $0$ loops is defined using spin $j$ monochromatic intertwiners. The scalar constraint operator adds and subtracts loops labelled with a given spin $l$. If $j=l$ the entire homogeneous-isotropic invariant space is described by monochromatic spin $j$ intertwiners only. The last example is our main motivation for this work. While in the previous two examples the focus is on 4-valent intertwiners, in our approach it is necessary to study the spaces of monochromatic intertwiners $\Inv{\Hil_j^{\otimes N}}$ with arbitrary valence $N$ (or at least in some range when the cut-off on the number of loops is used). We believe that spin $1/2$ spaces are of particular interest because they seem to describe the smallest homogeneous-isotropic grains of space, from which the quantum geometry of space is built (see \cite{RovelliBook} for the discussion of the quantum grains of space). 

In loop quantum gravity \cite{StatusReport,ThiemannBook,RovelliBook,Zakopane,LQG25,MaReview,AshtekarRovelliReuterRev,rovelli2014covariant} the volume operator plays an important role. It appears in the gravitational as well as in the matter part of the quantum scalar constraint operators. In loop quantum cosmology \cite{LQCI,LQCII,LQCIII,LQCIV} the volume operator is used to construct one of the two fundamental observables. An evolution of the operator in certain coherent states is calculated to show that the theory predicts Big Bounce instead of Big Bang -- the quantum universe contracts until it reaches a minimal volume at which it bounces and starts to expand. It is therefore unavoidable to study the volume operator when analyzing cosmological sector of loop quantum gravity.

There are two proposals for the volume operator: we will call them Rovelli-Smolin-DePietri volume operator \cite{RovelliSmolinVolume,DePietriRovelliVolume} and Ashtekar-Lewandowski volume operator \cite{AshtekarLewandowskiVolume}. The quantum operators correspond to the same classical object. In \cite{AshtekarLewandowskiVolume}, this fact is underlined in the nomenclature: they are called the volume operator in extrinsic and intrinsic regularization, respectively. The operators are defined in the spaces of intertwiners of valence at least 3 and vanish in the spaces of intertwiners of valence 3 \cite{LollThreeValent}. In the case of 4-valent intertwiners the Rovelli-Smolin-DePietri volume operator and Ashtekar-Lewandowski volume operator coincide. In this simplest, non-trivial case some spectral properties of the volume operator were studied analytically \cite{LollVolume,DePietriFourValent,BrunnemannI,BrunnemannIII}. In the case of Rovelli-Smolin-DePietri volume operator and higher valent monochromatic intertwiners, bounds were found for the eigenvalues \cite{MajorSeifert} but more detailed spectral properties in the higher valent cases are usually studied numerically. Our numerical study revealed a property not reported previously in the literature -- the volume operator is proportional to identity matrix in the spaces of spin $1/2$ monochromatic intertwiners. In this paper we prove it analytically and calculate the proportionality factor. The idea is that the Rovelli-Smolin-DePietri volume operator commutes with the permutations of the indices of the invariant tensors. It turns out that the space of spin $1/2$ monochromatic intertwiners $\Inv{\Hil_{\frac{1}{2}}^{\otimes N}}$ is an irreducible representation space of the group of permutations of $N$ elements, denoted here by $S_N$. The property of the volume operator follows from Schur's lemma. In this paper we will give even more general result -- we will find a decomposition of the space of spin $j$ monochromatic intertwiners into subspaces invariant under the action of the volume operator. Our main tool will be techniques used in the context of Schur-Weyl duality as presented in \cite{FultonHarris}. Our results explain to some extent the pattern of degeneracy of eigenstates of the Rovelli-Smolin-DePietri volume operator in the spaces on monochromatic spin $j$ intertwiners.
\section{The volume operator and the space of monochromatic intertwiners}
The volume operator corresponds to classical observable which is a volume of a region $R$ of the space manifold $\Sigma$. There are two proposals for the quantum operator, which we will call Rovelli-Smolin-DePietri volume operator \cite{RovelliSmolinVolume,LollVolume,DePietriRovelliVolume} and Ashtekar-Lewandowski volume operator \cite{AshtekarLewandowskiVolume}. The difference between the operators is just a quantization ambiguity -- they correspond to the same classical object. In \cite{AshtekarLewandowskiVolume} they are called, respectively, the volume operator in the extrinsic and intrinsic regularization. Our result will apply only to the former proposal. In this paper we will call it simply the volume operator.

We start this section by recalling the definition of the Loop Quantum Gravity Hilbert space and both versions of the volume operator (see \sref{sc:def_volume}). In \sref{sc:spin_networks} we recall the definition of the spin networks and discuss the action of the volume operator on the spin-network states. This action is described by a family of operators, each acting in the space of SU(2) invariant tensors (intertwiners) $\Inv{\Hil_{j_1}\otimes\ldots \otimes \Hil_{j_N}}$. In this paper we restrict to monochromatic intertwiners, i.e. to elements of $\Inv{\Hil_{j}^{\otimes N}}$. In \sref{sc:monochromatic_intertwiners} we discuss the action of the permutation group $S_N$ on $\Inv{\Hil_{j}^{\otimes N}}$ and show that the volume operator commutes with this action. The results presented in this paper are based on this observation.

\subsection{Definition of the volume operator}\label{sc:def_volume}
The Hilbert space of LQG is a space of cylindrical functions on the space of SU(2) connections on $\Sigma$ equipped with the Ashtekar-Lewandowski measure \cite{StatusReport,ThiemannBook}. In order to define the Hilbert space one considers oriented graphs $\gamma$ embedded in the space manifold $\Sigma$. A graph $\gamma$ is a collection of oriented compact 1-dimensional submanifolds called links $\links{\gamma}$ \footnote{The submanifolds need to satisfy certain properties, see \cite{StatusReport, LewandowskiSahlmann}.} meeting at the nodes $\nodes{\gamma}$. A complex function $\Psi$ on the space of SU(2) connections on $\Sigma$ is called cylindrical if there is a graph $\gamma$ and a complex function $\psi:{\rm SU(2)}^{N}\to \mathbb{C}$ such that
\be
\Psi(A)=\psi(A_{\link_1},\ldots, A_{\link_N}), 
\ee
where $\link_1,\ldots, \link_N$ are all links of the graph $\gamma$ and $A_\link\in {\rm SU(2)}$ is the parallel transport of the connection $A$ along the link $\link$. We will say that the function $\Psi$ is cylindrical with respect to the graph $\gamma$. Given two graphs $\gamma_1,\gamma_2$ and functions $\Psi_1, \Psi_2$ cylindrical with respect to $\gamma_1$ and $\gamma_2$, respectively, it is always possible to find a graph $\gamma$ such that $\Psi_1$ and $\Psi_2$ are cylindrical with respect to $\gamma$. The scalar constraint between $\Psi_1$ and $\Psi_2$ is defined by
\be\label{eq:scalar_product}
\braket{\Psi_1}{ \Psi_2}=\int_{{\rm SU(2)}^N} dg_1 \ldots dg_N \overline{\Psi_1}(g_1,\ldots, g_N) \Psi_2(g_1,\ldots, g_N),
\ee
where $N$ is the number of links in $\gamma$. The kinematical Hilbert space of LQG is the Cauchy completion of the space of cylindrical functions with respect to the scalar product \eqref{eq:scalar_product}.

The volume operators can be defined using 'angular momentum like' operators $J_{\node, I}^i$:
\be\label{eq:Joperator_definition}
\fl ( J_{\node, I}^i \Psi)(A)=\cases{\iu \frac{d}{dt}\big|_{t=0}\psi(A_{\link_1}, \ldots , A_{\link_I} \exp(t \tau^i), \ldots, A_{\link_N}),& if $\link_I$ is outgoing from $\node$,\\
\iu \frac{d}{dt}\big|_{t=0}\psi(A_{\link_1}, \ldots , \exp(-t \tau^i) A_{\link_I} , \ldots, A_{\link_N}),& if $\link_I$ is incoming to $\node$.}
\ee
Here, $\tau^i$ are su(2) matrices related to the Pauli matrices by $\tau^i=\frac{1}{2 \iu} \sigma^i$. Given a region $R$ in $\Sigma$ the action of the volume operator $\Volume_R$ on function cylindrical with respect to a graph $\gamma$ is a sum over operators, each depending only on small neighbourhood of the corresponding node of the graph:
\be
\Volume_R \Psi = \kappa_0 \left(\frac{8\pi G \hbar \gamma}{c^3}\right)^{\frac{3}{2}} \sum_{\node\in \nodes{\gamma}\cap R} \sqrt{|\hat{q}_\node|} \Psi.
\ee
Both proposals for the volume operator share the general form but differ in the choice of operators $\hat{q}_\node$.
\begin{enumerate}
\item In Ashtekar-Lewandowski volume operator the operator is:
\be
\hat{q}_\node \Psi = \frac{1}{8} \sum_{I=1}^{N_\node} \sum_{J=I}^{N_\node} \sum_{K=J}^{N_\node} \epsilon(\dot{\link}_I, \dot{\link}_J,\dot{\link}_K) \epsilon_{ijk} J^i_{\node,I} J^j_{\node,J} J^k_{\node,K},
\ee
where $\epsilon(\dot{\link}_I, \dot{\link}_J,\dot{\link}_K)=\sgn (\omega(\dot{\link}_I, \dot{\link}_J,\dot{\link}_K) ),$ $\omega$ is the orientation 3-form on $\Sigma$ and $\dot{\link}_I,\dot{\link}_J, \dot{\link}_K$ are vectors tangent to the links $\link_I,\link_J, \link_K$ at the node $\node$. 
\item In the Rovelli-Smolin-DePietri volume operator the operator is:
\be
\hat{q}_\node \Psi = \frac{1}{8} \sum_{I=1}^{N_\node} \sum_{J=I}^{N_\node} \sum_{K=J}^{N_\node} |\epsilon_{ijk} J^i_{\node,I} J^j_{\node,J} J^k_{\node,K}|.
\ee
Let us notice that the operator does not depend on the embedding of the graph in $\Sigma$.
\end{enumerate} 
\subsection{Spin networks}\label{sc:spin_networks}
The kinematical Hilbert space of Loop Quantum Gravity is spanned by the gauge-variant spin-network states \footnote{Typically a spin-network state is gauge invariant (see for example \cite{BaezSN,SFLQG}). On the other hand in the standard text-books \cite{StatusReport, ThiemannBook} the spin networks that are not gauge invariant are considered. We follow the nomenclature from \cite{ThiemannBook} and use a notion of gauge-variant spin networks.}. A gauge-variant spin network $s$ is a triple $(\gamma,\rho,\iota)$:
\begin{itemize}
\item an oriented graph $\gamma$ embedded in the space manifold $\Sigma$,
\item a labelling $\rho$ of the links of $\gamma$ with unitary irreducible representations of the SU(2) group:
\be
\links{\gamma} \ni \link \mapsto \rho_{\link} \in \GL{\Hil_{\link}},
\ee
where $\GL{V}$ is the space of linear automorphisms of $V$, $\Hil_{\link}$ is the Hilbert space in which the SU(2) representation acts,
\item a labelling $\iota$ of the nodes of $\gamma$ with tensors:
\be
\nodes{\gamma} \ni \node \mapsto \iota_\node \in \Hil_{\link_1}^*\otimes \ldots \Hil_{\link_M}^* \otimes \Hil_{\link_{M+1}} \otimes \ldots \otimes \Hil_{\link_N}, 
\ee
where the links $\link_1,\ldots ,\link_M$ are incoming to the node $\node$ and the links $\link_{M+1},\ldots, \link_{N}$ are outgoing from the node, $\Hil^*_\link$ denotes the Hilbert space dual to $\Hil_{\link}$.
\end{itemize}
A gauge-variant spin network such that each tensor $\iota_\node$ is invariant with respect to the action of the SU(2) group will be simply called a spin network. In this case we will write $\iota_\node\in \Inv{\Hil_{\link_1}^*\otimes \ldots \Hil_{\link_M}^* \otimes \Hil_{\link_{M+1}} \otimes \ldots \otimes \Hil_{\link_N}}$. A subspace spanned by all spin networks is the space of solutions of the Gauss constraint.

Each gauge-variant spin network $s$ defines a complex function $\psi_s:{\rm SU(2)}^{N}\to \mathbb{C}$ by a unique contraction \cite{SFLQG}:
\be\label{eq:spinnetwork_contraction}
\psi_s(u_{\link_1},\ldots, u_{\link_N})=\left( \bigotimes_{\link\in\links{\gamma}} \rho_\link(u_\link)\right) \lrcorner \left( \bigotimes_{\node\in\nodes{\gamma}} \iota_\node \right), \quad u_{\link_i}\in \SU.
\ee
The corresponding gauge-variant spin-network state is
\be
\Psi_s(A)=\psi_s(A_{\link_1},\ldots, A_{\link_N}).
\ee
The action of the operator $\hat{q}_\node$ on a gauge-variant spin-network state $\Psi_s$ is defined by its action on the space of tensors associated to the node $\node$. Let us focus on a single node $\node$ and assume for simplicity that all links are outgoing from $\node$. By applying the definition of the operator $ J_{\node, I}^i $ \eqref{eq:Joperator_definition} to the spin-network contraction \eqref{eq:spinnetwork_contraction} we obtain:
\be
(J_{\node, I}^i \Psi_s)(A)=\Psi_{s'}(A),
\ee
where $s'=(\gamma',\rho',\iota')$ is the same as $s=(\gamma,\rho,\iota)$ except for the tensor $\iota'_\node$ which is related to $\iota_\node$ by the following relation:
\be
\iota'_\node= (\id \otimes\ldots \otimes \id\otimes \rho'_{I}(\tau^i)\otimes \id \otimes \ldots \otimes \id) \cdot \iota_\node.
\ee 
In the formula above $\rho_I$ equals $\rho_{\link_{\node,I}}$ for the unique link $\link_{\node,I}$ defined by the pair $(\node,I)$ and 
\be
{\rm su}(2)\ni \tau \mapsto \rho'(\tau)= \frac{d}{dt}\big|_{t=0} \rho(\exp(t \tau)) \in {\rm End}(\Hil_{\link_{\node,I}})
\ee
is the representation of the su(2) Lie algebra induced by the representation $\rho$ of $\SU$. We will denote by $J^i_I$ a linear operator in the space $\Hil_{\rho_1}\otimes\ldots\otimes \Hil_{\rho_N}$ defined by:
\be
J^i_I =  \id \otimes\ldots \otimes \id\otimes \rho'_{I}(\tau^i)\otimes \id \otimes \ldots \otimes \id.
\ee
In the following, for each spin $j$ we will fix a representation of the $\SU$ group (a common choice is to take the Wigner $D$-matrices):
\be 
\rho_j:\SU\to \GL{\Hil_j}.
\ee
The Rovelli-Smolin-DePietri volume operator is defined by a family of operators $V_{\vec{j}},\ \vec{j}=\{j_1,\ldots,j_N\}$ each defined on the space $\Hil_{j_1}\otimes \ldots \otimes \Hil_{j_N}$:
\be
V_{\vec{j}} = \kappa_0 \left(\frac{8\pi G \hbar \gamma}{c^3}\right)^{\frac{3}{2}}\sqrt{|\hat{q}_{\vec{j}}|},
\ee
where
\be
\hat{q}_{\vec{j}}=\frac{1}{8} \sum_{I=1}^{N} \sum_{J=I}^{N} \sum_{K=J}^{N} |\epsilon_{ijk} J^i_{I} J^j_{J} J^k_{K}|.
\ee
It is straightforward to check that $V_{\vec{j}}$ commutes  with the left action of the $\SU$ group:
\be
V_{\vec{j}}\, \rho_{j_1}(u) \otimes \ldots \otimes \rho_{j_N}(u)  = \rho_{j_1}(u) \otimes \ldots \otimes \rho_{j_N}(u)\, V_{\vec{j}} \quad \forall u\in \SU.
\ee
As a result, the space $\Inv{\Hil_{j_1}\otimes\ldots\otimes\Hil_{j_N}}$ is an invariant space of the operator and $V_{\vec{j}}$ descends to an operator on $\Inv{\Hil_{j_1}\otimes\ldots\otimes\Hil_{j_N}}$.
\subsection{The space of monochromatic intertwiners}\label{sc:monochromatic_intertwiners}
In this paper we will study the volume operator in the spaces of monochromatic intertwiners. Following \cite{MajorSeifert} we will say that a space of intertwiners $\Inv{\Hil_{j_1}\otimes \ldots \otimes \Hil_{j_N}}$ is monochromatic if all spins are equal $j_1=j_2=\ldots=j_N=j$. 

There is a natural right action of the permutation group $S_N$ in the space of monochromatic intertwiners. Consider the right action of $S_N$ in the space $\Hil_j^{\otimes N}$:
\be
(v_1 \otimes \ldots \otimes v_N) \cdot \sigma = v_{\sigma(1)}\otimes \ldots \otimes v_{\sigma(N)}.
\ee
This action commutes with the left action of the $\GL{\Hil_j}$ group. In particular for any $u\in \SU$ and $\iota\in \Hil_j^{\otimes N}$ we have:
\be
\rho_j(u)^{\otimes N} \cdot  \left(\iota \cdot \sigma \right)=\left(\rho_j(u)^{\otimes N} \cdot \iota \right) \cdot \sigma.
\ee
For $\iota \in \Inv{\Hil_j^{\otimes N}}$ we have $\rho_j(u)^{\otimes N} \cdot \iota =\iota$ and therefore:
\be
\rho_j(u)^{\otimes N} \cdot  \left(\iota \cdot \sigma \right) =\iota \cdot \sigma.
\ee
This means that if $\iota\in  \Inv{\Hil_j^{\otimes N}}$ then $\iota\cdot \sigma \in  \Inv{\Hil_j^{\otimes N}}$. Let us notice that the resulting representation of the group $S_N$ in $\Inv{\Hil_j^{\otimes N}}$ is unitary (in the canonical scalar product pulled-back with the canonical embedding $\Inv{\Hil_j^{\otimes N}}\subset \Hil_j^{\otimes N}$).

The operators $J^i_I$ are covariant under the action of the group $S_N$:
\be
J^i_I \left((v_1 \otimes \ldots \otimes v_N) \cdot \sigma\right)=\left(J^i_{\sigma(I)} (v_1 \otimes \ldots \otimes v_N)\right)\cdot \sigma.
\ee
It is straightforward to verify that the operator $V_{\vec{j}}$ is invariant under this action:
\begin{equation}
V_{\vec{j}} \left((v_1 \otimes \ldots \otimes v_N) \cdot \sigma\right)=\left(V_{\vec{j}} (v_1 \otimes \ldots \otimes v_N)\right)\cdot \sigma.\label{eq:volume_invariance}
\end{equation}
Indeed, if we write $\hat{q}_{\vec{j}}$ in an equivalent form:
\be
\hat{q}_{\vec{j}}=\frac{1}{48} \sum_{I\neq J, I\neq K, J\neq K } |\epsilon_{ijk} J^i_{I} J^j_{J} J^k_{K}|,
\ee
we obtain:
\begin{eqnarray}
\fl \hat{q}_{\vec{j}} \left((v_1 \otimes \ldots \otimes v_N) \cdot \sigma\right)=\left( \frac{1}{48}\sum_{I\neq J, I\neq K, J\neq K }|\epsilon_{ijk} J^i_{\sigma(I)} J^j_{\sigma(J)} J^k_{\sigma(K)}|\cdot (v_1 \otimes \ldots \otimes v_N)\right)\cdot \sigma =\nonumber\\= \left(\hat{q}_{\vec{j}}  (v_1 \otimes \ldots \otimes v_N)\right)\cdot \sigma.
\end{eqnarray}
In the last equality we used the fact the $\sigma$ is a bijection and we changed the sum to be over $I'=\sigma(I),\, J'=\sigma(J),\, K'=\sigma(K) $. Let us notice that the invariance of $V_{\vec{j}}$ under the action of the permutation group holds only for the Rovelli-Smolin-DePietri volume operator.

The property \eqref{eq:volume_invariance} can be expressed by saying that $V_{\vec{j}}$ is an element of the commutator algebra (as defined in \cite{FultonHarris}):
\be
B=\{\varphi\in\End{U} : \varphi(v\cdot g)=\varphi(v)\cdot g, \forall v\in U, g\in G\},
\ee
where $U=\Inv{\Hil_j^{\otimes N}}, G=S_N$. In order to understand this notion better, let us consider some examples. If $U$ is an irreducible (right) representation space of $G$, $d=\dim U$, then by Schur's lemma it follows that
\be
B=\{ c \id_{d}: c\in \mathbb{C}\} \cong \mathbb{C}.
\ee
In similar manner if $U=U_i^{\oplus n_i}$, where $U_i$ is an irreducible (right) representation space of $G$, then (see \cite{Sagan} for more details)
\be
B=\{ m \otimes \id_{d_i}: m\in {\rm M}_{n_i}(\mathbb{C})\} \cong {\rm M}_{n_i}(\mathbb{C}),
\ee
where ${\rm M}_{n_i}(\mathbb{C})$ denotes the ring of $n_i\times n_i$ complex matrices. In general if $U=\bigoplus_i U_i^{\oplus n_i}$ is the decomposition of $U$ into irreducible subspaces, then
\be
B\cong \bigoplus_i {\rm M}_{n_i}(\mathbb{C}).
\ee

\section{Decomposition of the space of monochromatic intertwiners}
The main tool in this paper is the decomposition of the space of monochromatic intertwiners into $B$-irreducible spaces, where $B$ is the commutator algebra introduced in \sref{sc:monochromatic_intertwiners}.  In the basis adapted to this decomposition $V_{\vec{j}}$ has a block diagonal form. This property can be used to simplify the diagonalization problem. 

We will find the decomposition by applying some techniques described in \cite{FultonHarris} in the context of Schur-Weyl duality, which relates the finite dimensional representations of general linear and permutation groups. The representation theory of the permutation group is typically formulated using the language of group algebra and its representations rather than the group and its representations \cite{FultonHarris,Sagan}. Therefore we will introduce first the concept of the group algebra $\CG$ (see \sref{sc:group_algebra}). The irreducible representation spaces of the permutation group are labelled with the Young diagrams, which will be discussed in \sref{sc:Young_diagrams}. In \sref{sc:permutation_group_representations} we will discuss the representations of the permutation group. The decomposition will be presented in \sref{sc:decomposition} and in \sref{sc:application_of_decomposition} we will discuss its consequences for the eigenvalue problem. 
\subsection{The group algebra}\label{sc:group_algebra}
Let us assume that the group $G$ is finite, for example $G=S_N$. A group algebra $\CG$ is a vector space with basis $e_g, g\in G$:
\be
\CG=\{c_1 e_{g_1} + \ldots + c_n e_{g_n}: c_i \in\mathbb{C}, G=\{g_1,\ldots g_n\} \}
\ee
equipped with multiplication such that:
\be
 e_{g_1} \cdot e_{g_2} = e_{g_1 g_2}.
\ee
A representation $\tilde{\rho}$ of a group algebra $\CG$ on a vector space $V$ is an algebra homomorphism:
\be
\tilde{\rho}:\CG\to {\rm End}(V). 
\ee
A representation $\rho:G\to{\rm GL}(V)$ of a group $G$ extends by linearity to a representation $\tilde{\rho}:\CG\to {\rm End}(V)$ of the group algebra $\CG$. Any statement about a representation of a group $G$ has an equivalent statement in the group algebra $\CG$. 

A left action of $\CG$ on $V$ makes it a left $\CG$-module. For completeness we will recall the definition of a module. Let $A$ be a ring. An Abelian group $U$ is a left $A$-module if there is a mapping $A\times U \to U, (a,u)\mapsto a u$ such that:
\begin{itemize}
\item $a(u_1 + u_2)= a u_1 + a u_2$,
\item $(a_1 + a_2) u=a_1 u + a_2 u$,
\item $a_1 (a_2 u) = (a_1 a_2) u$,
\item $1 u = u$.
\end{itemize}
Clearly $\CG$ is a ring (because any algebra is a ring), $V$ is an Abelian group (with respect to the addition of vectors) and the action of $\CG$ on $V$ satisfies the axioms above.

We can make $V$ also a right $\CG$-module using an anti-involution:
\be 
a=\sum_{g\in G} a_g e_g \mapsto \hat{a}=\sum_{g\in G} a_g e_{g^{-1}}.
\ee
Left $\CG$-module is turned into right $\CG$-module by taking:
\be
v \cdot \hat{a}=a\cdot v. 
\ee

A (left) regular representation of a finite group $G$ is a group homomorphism 
\be
\rho:G\to {\rm GL}(\CG)
\ee
given by
\be
\rho(g) e_{g_i}=e_{g g_i}. 
\ee
In the language of group algebras: the regular representation is a left $\CG$-module over $\CG$. The regular representation $R$ can be decomposed into its irreducible components $W_i$. Each representation space $W_i$ appears $\dim W_i$ times:
\be
R =\bigoplus_i W_i ^{\oplus \dim W_i}.
\ee 
Each irreducible representation is isomorphic to a minimal left ideal in $\CG$. Let us recall that a (left) ideal in a ring $A$ is a subring $I$ closed under (left) multiplication by elements of the ring $A$: $A I = I$. Each (left) ideal in $\CG$ is generated by an idempotent. This means in particular that for each space $W_i$ there is $e_i\in A$ such that $e_i^2 = e_i$ and
\be
W_i = A e_i.
\ee We will present a construction of such idempotents for the permutation group $S_d$, called Young symmetrizers. We will introduce first a convenient tool called Young diagrams and Young tableaux. 
\subsection{Young diagram and Young tableau}\label{sc:Young_diagrams}
A Young diagram is associated to a partition 
\be 
(\lambda_1, \ldots , \lambda_k), \quad \lambda_1 \geq \lambda_2 \geq \ldots \geq \lambda_k \geq 1, \lambda_i \in \mathbb{N}
\ee
of a number $d\in \mathbb{N}$:
\be
d=\lambda_1 +\ldots + \lambda_k. 
\ee
A Young diagram has $\lambda_i$ boxes in the $i$-th row. For example, the Young diagram corresponding to the partition $12=5+3+3+1$ is:
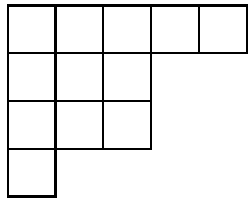
\begin{center}
\ydiagram{5, 3, 3, 1}
\end{center}
The conjugate partition $\lambda'=(\lambda'_1,\ldots ,\lambda'_r)$ is obtained by interchanging rows and columns in the Young diagram. For example, the Young diagram 
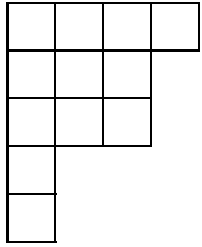
\begin{center}
\ydiagram{4, 3, 3, 1,1}
\end{center}
corresponds to the partition conjugate to $12=5+3+3+1$. A Young tableau is obtained from a Young diagram $\lambda$ by inserting numbers $1,\ldots, d$ into the boxes. Following \cite{Sagan} we will call a $\lambda$-tableau any Young tableau on the Young diagram corresponding to a partition $\lambda$. For example:
\begin{equation}
\label{eq:canonical_tableau}
\begin{ytableau}
1 & 2 & 3 &4 &5 \\
6&7&8\\
9&10&11\\
12
\end{ytableau}
\end{equation}
is a Young tableau on the Young diagram $\lambda=(5,3,3,1)$. Another Young tableau on this diagram could be:
\be\label{eq:example_tableau}
\begin{ytableau}
1 & 2 & 6 &7 &5 \\
10&11&8\\
3&4&9\\
12
\end{ytableau}
\ee
Young tableau is called canonical if the boxes are labelled by consecutive numbers. For example the tableau \eqref{eq:canonical_tableau} is the canonical Young tableau on the Young diagram $\lambda=(5,3,3,1)$.  Young tableau is called standard if all rows and columns are increasing. Each canonical tableau is standard. Another example of standard tableau on this diagram is:
\be
\begin{ytableau}
1 & 2 & 5 &6 &7 \\
3&4&8\\
9&10&12\\
11
\end{ytableau}
\ee
The tableau \eqref{eq:example_tableau} is not standard. 
\subsection{Representations of the permutation group}\label{sc:permutation_group_representations}
The irreducible representations of the permutation group $S_d$ can be constructed using the Young symmetrizers. For a given Young tableau $t^\lambda$, two subgroups of $S_d$ are constructed
\be
P(t^\lambda) = \{ g\in S_d: g\textrm{ preserves each row of } t^\lambda\}
\ee
and
\be
Q(t^\lambda) = \{ g\in S_d: g\textrm{ preserves each column of } t^\lambda\}.
\ee
To these subgroups there correspond two elements in the group algebra $\C{S_d}$:
\be
a(t^\lambda)=\sum_{g\in P(t^\lambda)} e_g,\quad b(t^\lambda)=\sum_{g\in Q(t^\lambda)} \sgn(g) e_g.
\ee
A Young symmetrizer is a product of these elements:
\be
c(t^\lambda)=a(t^\lambda) \cdot b(t^\lambda) \in \C{S_d}. 
\ee
Let $t^\lambda_{\rm can}$ be the canonical $\lambda$-tableau (let us recall that in the canonical Young tableau the boxes of the Young diagram $\lambda$ are labelled by consecutive numbers). The following notation is used:
\be
\fl P_\lambda= P(t^\lambda_{\rm can}),\quad Q_\lambda= Q(t^\lambda_{\rm can}),\quad a_\lambda = a(t^\lambda_{\rm can}),\quad b_\lambda=b(t^\lambda_{\rm can}),\quad c_\lambda=c(t^\lambda_{\rm can}).
\ee
A central role is played by the following theorem (see Theorem 4.3 in \cite{FultonHarris}).
\begin{theorem}
The Young symmetrizer $c_\lambda$ is an idempotent up to a scalar factor:
\be
c_\lambda^2 = n_\lambda c_\lambda. 
\ee
The image of $c_\lambda$ by right multiplication on $\C{S_d}$ is an irreducible representation $V_\lambda$ of $S_d$. Every irreducible representation can be obtained this way for a unique partition.
\end{theorem}

The number $n_\lambda$ in the theorem is given by
\be
n_\lambda=\frac{d!}{\dim V_\lambda}. 
\ee
There is a convenient expression for the dimension of the space $V_\lambda$, called the hook length formula. The hook length of a box in a Young diagram is the number of boxes directly below or directly to the right from the box, including the box once. For example, in the following diagram each box is labelled by its hook length:
\be
 \begin{ytableau}
8 & 7 & 6 &2 &1 \\
5&3&2\\
4&2&1\\
1
\end{ytableau}
\ee
The dimension of the space $V_\lambda$ is given by
\be\label{eq:dimVlambda}
\dim V_\lambda = \frac{d!}{\prod \textrm{(Hook lengths)}}.
\ee
For the Young diagram $\lambda=(5,3,3,1)$ above the formula becomes:
\be
\dim V_{\lambda }=\frac{12!}{8\cdot 7 \cdot 6 \cdot 2 \cdot 5 \cdot 3 \cdot 2 \cdot 4 \cdot 2}.
\ee
From the theorem follows that $\C{S_d}$ has the following decomposition:
\be\label{eq:CSd_splitting}
\C{S_d}\cong\bigoplus_\lambda V_\lambda^{\oplus \dim V_\lambda},
\ee
where the external direct sum is over all partitions of $d$ and $V_\lambda=A c_\lambda, \, A=\C{S_d}$. It turns out that the dimension of $V_\lambda$ is also equal to the number of standard $\lambda$-tableaux (for the definition of the standard tableau see the previous subsection). In fact, $A=\C{S_d}$ is a direct sum of left ideals $V_{T_\lambda}=A e_{T_\lambda}$, where 
\be\label{eq:idempotents}
e_{T_\lambda}=\frac{\dim V_\lambda}{ d!} c_{T_\lambda}
\ee and $T_\lambda$ runs through all the standard $\lambda$-tableaux:
\be\label{eq:CSd_splitting_equality}
 \C{S_d}=\bigoplus_\lambda \bigoplus_{T_\lambda} V_{T_\lambda}.
\ee
Each element $a\in \C{S_d}$ can be uniquely written as
\be\label{eq:a_splitting}
a =\sum_\lambda \sum_{T_\lambda} a_{T_\lambda}, 
\ee
where $a_\lambda = a e_{T_\lambda} \in V_{T_\lambda}$. Because of the property \eqref{eq:a_splitting}, we used an equality sign in the equation \eqref{eq:CSd_splitting_equality}, while in \eqref{eq:CSd_splitting} we used $\cong$ sign denoting an isomorphism. Let us notice that for all standard Young tableaux on a diagram $\lambda$ the spaces $V_{T_\lambda}$ are isomorphic to $V_\lambda$: $V_{T_\lambda}\cong V_\lambda$.

\subsection{Decomposition of the space of intertwiners}\label{sc:decomposition}
Fundamental role in our approach is played by Lemma 6.22 from \cite{FultonHarris}. In \cite{FultonHarris} it is used in the context of Schur-Weyl duality. We will be interested in this lemma for its own sake.

We will focus on the group of permutations but in the following lemma, $G$ is any finite group. Let us recall that a map $\varphi:U\to V$ is called a homomorphism of right modules if it is a homomorphism of the Abelian groups $U, V$ and satisfies
\be
\varphi(u\cdot a)=\varphi(u)\cdot a.
\ee
If $U$ is a right-module over $A=\C{G}$, the commutator algebra $B$ is the algebra of automorphisms of $U$. Let us recall that in \sref{sc:monochromatic_intertwiners} we used an equivalent definition:
\be
B={\rm Hom}_G(U,U)=\{\varphi:U\to U ; \varphi(v\cdot g)=\varphi(v)\cdot g, \forall v\in U, g\in G\}. 
\ee

Let $W$ be any left $A$-module. Consider a submodule $N$ of a tensor product $U\otimes_{\mathbb{C}}W$ generated by $\{ v a \otimes w - v\otimes a w\}$. The tensor product:
\be
U\otimes_A W = \left( U\otimes_{\mathbb{C}}W\right) \slash N
\ee
is a left $B$-module:
\be
b\cdot (v\otimes w) = (b\cdot v) \otimes w.
\ee
Let us notice that $U\otimes_A A$ and $U$ are isomorphic. The canonical isomorphism $\pi:U\otimes_A A \to U$ is
\be\label{eq:canonical_isomorphism}
\pi(v\otimes a)=v\cdot a.
\ee
\begin{lemma}\label{lm:lemma_decomposition}
Let $U$ be a finite-dimensional right $A$-module.
\begin{enumerate}
\item For any $c\in A$ the restriction of the canonical map $\pi$ to $U\otimes_A A c$  is an isomorphisms of left $B$-modules $U\otimes_A A c$ and $U c$.
\item\label{lm:UWirreducible} If $W=Ac$ is an irreducible left $A$-module, then $U\otimes_A W$ is an irreducible left $B$-module.
\item\label{lm:U_splitting} If $W_i=A c_i$ are the irreducible left $A$-modules in the decomposition $A\cong \bigoplus_i W_i^{\oplus m_i},$ then 
\be
U\cong \bigoplus_i( U \otimes_A W_i)^{\oplus m_i} \cong \bigoplus_i( U c_i)^{\oplus m_i}
\ee
is the decomposition of $U$ into irreducible left $B$-modules.
\end{enumerate}
\end{lemma}
The proof of this lemma can be found in \cite{FultonHarris}. We will not repeat it here. 

As we have shown in \sref{sc:monochromatic_intertwiners} the space of invariant tensors $U=\Inv{\Hil_j^{\otimes N}}$ is a right module over $A=\C{S_N}$. We have discussed above that $A=\C{S_N}$ has the following decomposition:
\be
A= \bigoplus_\lambda \bigoplus_{T_\lambda} A e_{T_\lambda},
\ee
where the direct sum is over all standard Young tableaux on the Young diagrams with $N$ boxes, $e_{T_\lambda}$ are idempotents from \eqref{eq:idempotents} constructed from Young symmetrizers. This decomposition leads to a decomposition of $U\otimes_A A$
\be
U\otimes_A A = \bigoplus_\lambda \bigoplus_{T_\lambda} U\otimes_A A e_{T_\lambda}.
\ee
By lemma \ref{lm:lemma_decomposition} (\ref{lm:UWirreducible}) each component in this decomposition $U\otimes_A A e_{T_\lambda}$ is an irreducible left $B$-module. By using the canonical isomorphism $\pi$ from \eqref{eq:canonical_isomorphism} we have
\be
U=\bigoplus_\lambda \bigoplus_{T_\lambda} U e_{T_\lambda}.
\ee
Each element $u\in U$ has the following decomposition 
\be
u=\sum_{\lambda} \sum_{T_\lambda} u_{T_\lambda},
\ee
where $u_{T_\lambda} = u e_{T_\lambda} $ is the component of $u$ in the irreducible left $B$-module $ U e_{T_\lambda}$. Let us summarize the construction by the following theorem.

\begin{theorem}\label{thm:HilbertSpaceDecomposition}
Let $\Inv{\Hil_j^{\otimes N}}$ be the space of tensors in $\Hil_j^{\otimes N}$ invariant under the action of the SU(2) group and $S_N$ be the group of permutations of $N$ elements.
\begin{enumerate}
\item The subspace $\Inv{\Hil_j^{\otimes N}}$ is invariant under the right action of the group $S_N$ on $\Hil_j^{\otimes N}$ defined by:
\be
(v_1 \otimes \ldots \otimes v_N) \cdot \sigma = v_{\sigma(1)}\otimes \ldots \otimes v_{\sigma(N)},\quad v_i\in \Hil_j.
\ee
\item Consider the commutator algebra:
\be
\fl B=\{b:\Inv{\Hil_j^{\otimes N}}\to \Inv{\Hil_j^{\otimes N}}; b(\iota \cdot \sigma) =b(\iota)\cdot \sigma,\ \forall \iota\in\Inv{\Hil_j^{\otimes N}}, \sigma\in S_N  \}.
\ee
Let $T_\lambda$ be a standard $\lambda$-tableau and $e_{T_\lambda}$ be the idempotent corresponding to the Young symmetrizer $c_{T_\lambda}$:
\be
e_{T_\lambda}=\frac{\dim V_\lambda}{ N!} c_{T_\lambda}.
\ee
Let us denote by $\InvLambda{\Hil_j^{\otimes N}}$ the image of $e_{T_\lambda}$:
\be
\InvLambda{\Hil_j^{\otimes N}} = \Inv{\Hil_j^{\otimes N}} e_{T_\lambda}.
\ee
The space $\Inv{\Hil_j^{\otimes N}}$ has the following decomposition into irreducible left $B$-modules:
\be
 \Inv{\Hil_j^{\otimes N}} = \bigoplus_\lambda \bigoplus_{T_\lambda} \InvLambda{\Hil_j^{\otimes N}},
\ee
where $\lambda$ runs through all Young diagrams with $N$ boxes, $T_\lambda$ runs through all the standard $\lambda$-tableaux.
\end{enumerate} 
\end{theorem}
\subsection{Consequences of the decomposition of the space of intertwiners for the eigenvalue problem}\label{sc:application_of_decomposition}
The theorem has the following consequences. Consider a hermitian operator 
\be 
C:\Inv{\Hil_j^{\otimes N}} \to \Inv{\Hil_j^{\otimes N}}
\ee such that
\be\label{eq:Cinvariance}
C(\iota \cdot \sigma) = C(\iota) \cdot \sigma \quad \forall \iota \in \Inv{\Hil_j^{\otimes N}}, \sigma\in S_N.
\ee
An example of such operator is the loop quantum gravity volume operator. Clearly $C$ is an element of the commutator algebra ${\rm Hom}_{S_N}(\Inv{\Hil_j^{\otimes N}},\Inv{\Hil_j^{\otimes N}})$. Let us recall the definition of $\InvLambda{\Hil_j^{\otimes N}}$:
\be
\InvLambda{\Hil_j^{\otimes N}}=\Inv{\Hil_j^{\otimes N}}  e_{T_\lambda}.
\ee
From theorem \ref{thm:HilbertSpaceDecomposition} it follows that $C$ (in the basis adapted to the decomposition) has block diagonal form. Let us denote by 
\be 
C_{T_\lambda}:\InvLambda{\Hil_j^{\otimes N}}\to \InvLambda{\Hil_j^{\otimes N}}
\ee
the corresponding blocks. The diagonalization problem of $C$ reduces to the problem of diagonalizing the operators $C_{T_\lambda}$, in particular if $\dim \InvLambda{\Hil_j^{\otimes N}}=1$ we obtain an eigenvector of $C_{T_\lambda}$. 

Let us notice that from the property \eqref{eq:Cinvariance} it follows that for a given $\lambda$ all operators $C_{T_\lambda}$ are unitarily equivalent. Indeed, consider an action of the permutation group $\sigma$ on $\lambda$-tableau $t^{\lambda}$. Let us denote by $t^{\lambda}_{ij}$ the number in the box in the $i$-th row and $j$-th column in $\lambda$-tableau $t^{\lambda}$. The $\lambda$-tableau $\sigma\cdot t^{\lambda}$ has a number $\sigma(t^{\lambda}_{ij})$ in the box in the $i$-th row and $j$-th column:
\be
(\sigma\cdot t^{\lambda})_{ij}=\sigma(t^{\lambda}_{ij}).
\ee
Permuting the numbers in the boxes induces group isomorphism $P(\sigma t^{\lambda})\cong P(t^{\lambda})$ and $Q(\sigma t^{\lambda})\cong Q(t^{\lambda})$:
\be
P(\sigma t^{\lambda}) = \sigma \cdot P(t^{\lambda}) \cdot \sigma^{-1},\quad Q(\sigma t^{\lambda}) = \sigma \cdot Q(t^{\lambda}) \cdot \sigma^{-1}.
\ee
As a result, the Young symmetrizer transforms as:
\be
c(\sigma \cdot t^\lambda)= \sigma \cdot c(t^{\lambda}) \cdot \sigma^{-1}.
\ee
It follows that the spaces $\InvLambda{\Hil_j^{\otimes N}}$ and $\InvL{\Hil_j^{\otimes N}}{\sigma \cdot T_\lambda}$ are unitarily equivalent. The equivalence is given by an operator $U_\sigma$:
\be
U_\sigma:\InvLambda{\Hil_j^{\otimes N}} \to \InvL{\Hil_j^{\otimes N}}{\sigma \cdot T_\lambda},\quad U_\sigma\, \iota = \iota\cdot \sigma^{-1}.
\ee
From the property \eqref{eq:Cinvariance} it follows that
\be
C_{\sigma\cdot T_\lambda} = U_\sigma C_{T_\lambda} U_\sigma^{-1}.
\ee
\section{Monochromatic spin $1/2$ intertwiners}\label{sc:spin_half}
In this section we will show that the volume operator is proportional to identity in the spaces of spin $1/2$ monochromatic intertwiners $\Inv{\Hil_{\frac{1}{2}}^{\otimes N}}$. We will calculate the proportionality factor. 

If the valence $N$ of the intertwiner is odd, the intertwiner space is $0$. If it is even, it will turn out to be the given by the $n$-th Catalan number (see \sref{sc:dim_inv}) which is equal to the number of standard $\lambda$-tableaux on a diagram $\lambda$ with two rows of equal length: $\lambda_0=(N/2,N/2)$. This means that dimension of each space $\InvL{\Hil_{\frac{1}{2}}^{\otimes N}}{T_{\lambda_0}}$ is not greater than $1$. In \sref{sc:decomposition_inv} we will show that it is non-trivial and therefore $\dim \InvL{\Hil_{\frac{1}{2}}^{\otimes N}}{T_{\lambda_0}}=1$. This means that the elements of $\InvL{\Hil_{\frac{1}{2}}^{\otimes N}}{T_{\lambda_0}}$ are eigenvectors of the volume operator and that the volume operator is proportional to identity. In \sref{sc:TrV} we will calculate the proportionality factor by calculating the trace of the operator.
\subsection{Dimension of the space of invariant tensors}\label{sc:dim_inv}

Let us notice that the operator $P:\Hil_{\frac{1}{2}}^{\otimes N} \to \Hil_{\frac{1}{2}}^{\otimes N}$:
\be
P=\int_{{\rm SU(2)}} d\mu(g) \underbrace{\rho_{\frac{1}{2}}(g)\otimes\rho_{\frac{1}{2}}(g)\otimes \ldots \otimes \rho_{\frac{1}{2}}(g)}_{N}
\ee
is a projector onto the subpace of invariant tensors $\Inv{\Hil_{\frac{1}{2}}^{\otimes N}}$. The dimension $D_N = \dim \Inv{\Hil_{\frac{1}{2}}^{\otimes N}}$of the space of invariants is therefore
\be
D_N=\int d\mu(g) \left( \Tr{\rho_{\frac{1}{2}}(g)}\right) ^N.
\ee
Let us notice, that $D_N$ is non-zero if $N$ is even: $N=2n, n\in \mathbb{N}$. This leads to the following formula
\be
\fl D_{2n}=\frac{4}{\pi} \int_{0}^{\frac{\pi}{2}} d\theta \sin ^2 \theta \left( \frac{\sin(2\theta)}{\sin \theta} \right)^{2n} =\frac{4^{n+1}}{\pi} \int_0^{\frac{\pi}{2}} d\theta (\sin\theta)^2 (\cos \theta)^{2n}=\frac{4^{n+1}}{2 \pi} B(n+\frac{1}{2}, \frac{3}{2}),
\ee
where $B(x,y)$ is the beta function. Since $n$ is an integer, this expression simplifies to 
\be
D_{2n}=\frac{1}{n+1} \binom{2n}{n}.
\ee
This is precisely the $n$-th Catalan number.
\subsection{Decomposition of the space of invariant tensors}\label{sc:decomposition_inv}
In the following we will assume that $N=2n,\, n\in \mathbb{N}$. Let us calculate the dimension of the representation space $V_\lambda$ of the permutation group $S_N$ corresponding to Young diagram $\lambda=(n, n)$. For example if $N=20$, then $\lambda=(10,10)$ and the corresponding Young diagram is:
\be
\ydiagram{10, 10}.
\ee
The dimension of the space $V_\lambda$ can be conveniently calculated using the hook length formula:
\be
\dim V_\lambda= \frac{(2n)!}{(n+1)! n!}=\frac{1}{n+1} \binom{2n}{n}.
\ee
For example in the following diagram each box is labelled by its hook length:
\be
\begin{ytableau}
11 & 10 & 9 & 8 &7&6 &5&4&3&2 \\
10 & 9 & 8 &7&6 &5&4&3&2&1
\end{ytableau}
\ee
Therefore for this diagram
\be
\dim V_\lambda= \frac{20!}{11! 10!}.
\ee
We will use this observation to prove the following lemma:
\begin{lemma}
For $N$ odd the space $\Inv{\Hil_{\frac{1}{2}}^{\otimes N}}$ is trivial and for $N$ even it has the following decomposition into irreducible left $B$-modules:
\be
\Inv{\Hil_{\frac{1}{2}}^{\otimes N}}=\bigoplus_{T_\lambda} \InvLambda{\Hil_{\frac{1}{2}}^{\otimes N}},
\ee
where $\lambda=(N/2,N/2)$ and the direct sum is over all standard $\lambda$-tableaux. Each space $\InvLambda{\Hil_{\frac{1}{2}}^{\otimes N}}$ in this decomposition is 1-dimensional:
\be
\dim \InvLambda{\Hil_{\frac{1}{2}}^{\otimes N}} = 1.
\ee
\end{lemma}
\begin{proof}
Since $\dim  \Inv{\Hil_{\frac{1}{2}}^{\otimes N}} = \dim V_\lambda$ is the same as the number of standard $\lambda$-tableaux, it follows that
\be
\dim \InvLambda{\Hil_{\frac{1}{2}}^{\otimes N}} \leq 1, \ \rm{for}\ \lambda=(N/2,N/2).
\ee
In fact it is non-trivial. Consider an invariant antisymmetric tensor $\epsilon \in \Inv{\Hil_{\frac{1}{2}}\otimes \Hil_{\frac{1}{2}}}$:
\be
\tensor{\rho(u)}{^{A_1}_{B_1}}  \tensor{\rho(u)}{^{A_2}_{B_2}}\, \epsilon^{B_1 B_2}=\epsilon^{A_1 A_2}.
\ee
It is defined uniquely up to a scale factor. Let us for simplicity choose the scale factor in such a way that $\epsilon^{-\frac{1}{2} \frac{1}{2}}=-1$. A non-trivial element is:
\be 
\iota^{A_1 \ldots A_n A_{n+1}\ldots A_{2n}}=\epsilon^{(A_1 |A_{n+1}|} \epsilon^{A_2 |A_{n+2}|} \ldots \epsilon^{A_{n}) A_{2n}},
\ee
where the brackets $()$ denote symmetrization and the vertical lines distinguish the fixed indices (those not used in the symmetrization). In other words, take 
\be
\tilde{\iota}^{A_1 \ldots A_n A_{n+1}\ldots A_{2n}} = \epsilon^{A_1 A_{n+1}}\ldots \epsilon^{A_n A_{2n}}.
\ee The invariant tensor $\iota$ is obtained from $\tilde{\iota}$ by acting with the Young symmetrizer corresponding to the canonical Young tableau $T_\lambda=t^{\lambda}_{\rm can}$, where $\lambda=(N/2,N/2)$: 
\be 
\iota=\tilde{\iota}\cdot c_\lambda.
\ee Clearly $\iota$ is symmetric in the first $N/2$ and the last $N/2$ indices, while it is antisymmetric under transposition of indices $A_{i}$ and $A_{i+n}$. It is also invariant because $\epsilon$ is invariant. It is non-trivial, because:
\be
\iota^{-\frac{1}{2} \ldots -\frac{1}{2} \frac{1}{2}\ldots \frac{1}{2}}=(-1)^{n}.
\ee
This shows that 
\be
\dim \InvL{\Hil_{\frac{1}{2}}^{\otimes N}}{t^{\lambda}_{\rm can}} = 1\ \rm{for}\ \lambda=(N/2,N/2).
\ee
Since all the spaces $\InvLambda{\Hil_{\frac{1}{2}}^{\otimes N}}$ corresponding to the same Young diagram $\lambda$ are isomorphic to $\InvL{\Hil_{\frac{1}{2}}^{\otimes N}}{t^{\lambda}_{\rm can}}$, it follows that:
\be
\Inv{\Hil_{\frac{1}{2}}^{\otimes N}}=\bigoplus_{T_\lambda} \InvLambda{\Hil_{\frac{1}{2}}^{\otimes N}},
\ee
where $\lambda=(N/2,N/2)$ and the direct sum is over all standard $\lambda$-tableaux.
\end{proof}
It follows that the volume operator $V$ (and any element of the commutator algebra $B$) is diagonal in the basis adapted to the decomposition. From the discussion in \sref{sc:application_of_decomposition} it follows that the blocks $V_{T_\lambda}$ are isospectral. As a result, all the diagonal entries are equal. The volume operator restricted to the space $\Inv{\Hil_{\frac{1}{2}}^{\otimes N}}$ is proportional to identity. We will find the overall factor by calculating the trace of the operator.
\subsection{The trace of the volume operator}\label{sc:TrV}
The calculation of the trace of the volume operator is based on the observation that the operators $q_{IJK}:\Inv{\Hil_j^{\otimes N}}\to \Inv{\Hil_j^{\otimes N}}$,
\be
q_{IJK}=|\epsilon_{ijk} J^i_{I} J^j_{J} J^k_{K}|,\quad I\neq J,\, I\neq K,\, J\neq K
\ee
are isospectral:
\be\label{eq:UqU}
q_{\sigma(I) \sigma(J) \sigma(K)} = U_\sigma q_{IJK} U_{\sigma}^{-1},
\ee
where $U_\sigma$ is the unitary operator defined by the right action of the permutation group on $\Inv{\Hil_j^{\otimes N}}$:
\be
U_\sigma \iota = \iota \cdot \sigma^{-1}.
\ee
As a result, the problem of calculating the trace of the operator $q_{\vec{j}}$ reduces to the problem of calculating the trace of $q_{123}$:
\be\label{eq:Trqj}
\Tr{q_{\vec{j}}}=\frac{1}{8}\binom{N}{3} \Tr{q_{123}}.
\ee

In order to calculate $\Tr{q_{123}}$ in the space $\Inv{\Hil_{\frac{1}{2}}^{\otimes N}}, N=2n, n\in \mathbb{N}$, we notice that there is an isomorphism
\be
\mathcal{E}:\bigoplus_{k=\frac{1}{2}}^{\frac{3}{2}}\left(\Inv{\Hil_{\frac{1}{2}}^{\otimes 3} \otimes \Hil_k} \otimes \Inv{\Hil_k\otimes \Hil_{\frac{1}{2}}^{\otimes (N-3)}} \right)\to \Inv{\Hil_{\frac{1}{2}}^{\otimes N}},
\ee
defined by 
\be\label{eq:Eiso}
(\mathcal{E}(v\otimes w))^{A_1 \ldots A_N}=v^{A_1 A_2 A_3 B_1} \epsilon_{B_1 B_2} w^{B_2 A_4 \ldots A_N},
\ee
for all $ v\in  \Inv{\Hil_{\frac{1}{2}}^{\otimes 3} \otimes \Hil_k},\, w\in \Inv{\Hil_k\otimes \Hil_{\frac{1}{2}}^{\otimes (N-3)}}, k\in\{\frac{1}{2},\frac{3}{2}\}$. This isomorphism commutes with the operator $q_{123}$:
\be\label{eq:qEEq}
q_{123} \mathcal{E}(v\otimes w)=\mathcal{E}((q_{123} v)\otimes w).
\ee
In order to calculate $q_{123}$ it is enough to know its action in the space $\Inv{\Hil_{\frac{1}{2}}^{\otimes 3} \otimes \Hil_k}$ but this has been studied in \cite{BrunnemannI,BrunnemannIII,BrunnemannIV}. In particular, from formulas (7.3) and (7.4) from \cite{BrunnemannIV} it is straightforward to calculate that 
\be
q_{123} v = 0 \quad \forall v\in \Inv{\Hil_{\frac{1}{2}}^{\otimes 3} \otimes \Hil_{\frac{3}{2}}}.
\ee
The matrix of the operator $q_{123}$ in $\Inv{\Hil_{\frac{1}{2}}^{\otimes 4}}$ is a $2\times 2$ matrix, which turns out to be diagonal:
\be
q_{123} v = \frac{\sqrt{3}}{4} v \quad \forall v\in \Inv{\Hil_{\frac{1}{2}}^{\otimes 3} \otimes \Hil_{\frac{1}{2}}}.
\ee
The trace of the operator $q_{123}$ is therefore
\be
\Tr{q_{123}}=\frac{\sqrt{3}}{2}\cdot \dim \Inv{\Hil_{\frac{1}{2}}^{\otimes(N-2)}}=\frac{\sqrt{3}}{2} \frac{1}{n} \binom{2n-2}{n-1}.
\ee
From equation \eqref{eq:Trqj} it follows that
\be
\Tr{q_{\vec{j}}}=\frac{\sqrt{3}}{16}\frac{1}{n}  \binom{2n}{3}\binom{2n-2}{n-1}.
\ee
In order to calculate the factor $\lambda$, 
\be 
q_{\vec{j}}=\lambda \id
\ee we need to divide the trace by the dimension of the space $\Inv{\Hil_{\frac{1}{2}}^{\otimes(N)}}$:
\be
\lambda = \frac{\Tr{q_{\vec{j}}}}{\dim  \Inv{\Hil_{\frac{1}{2}}^{\otimes(N)}} }.
\ee
We obtain:
\be
\lambda = \frac{\sqrt{3}}{64\cdot 3!} (N-2)N(N+2).
\ee
Therefore for $\vec{j}=(\frac{1}{2},\ldots, \frac{1}{2})$ the volume operator is:
\be
V_{\vec{j}}=\frac{\kappa_0}{8} \left(\frac{8\pi G \hbar \gamma}{c^3}\right)^{\frac{3}{2}} \sqrt{\frac{\sqrt{3}}{3!} (N-2)N(N+2)}\cdot \id.
\ee
Let us notice that the total spin is $N\cdot \frac{1}{2} = \frac{N}{2}$. In the limit of large total spin the expected asymptotics \cite{MajorSeifert} is recovered:
\be
|| V_{\vec{j}} || \approx C (\textrm{total spin})^{\frac{3}{2}}.
\ee
\section{Higher spin monochromatic intertwiners}
In the case of spin $1/2$ monochromatic intertwiners we gave a full characterization of the volume operator. This was possible because the space of intertwiners turned out to be an irreducible space of the action of the permutation group. As a result, the volume operator is proportional to identity. In addition to this, the operator $q_{123}$ takes a simple form, which allowed us to calculate the overall factor. In general the space of intertwiners splits into subspaces corresponding to more than one Young diagram and an irreducible $B$-module can be more than $1$-dimensional. For the lowest spins and valences of the intertwiners our analysis predicts well the pattern of the degeneracy of the eigenvalues. In this section we will present some results of our numerical calculation of the dimensions of the irreducible spaces and the corresponding eigenvalues.

In our numerical calculation we use the tree basis of the intertwiner space $\Inv{\Hil_{j_1}\otimes \ldots \otimes \Hil_{j_N}}$. We recall the definition in \sref{sc:tree_basis} (see for example \cite{MakinenRecoupling} for detailed presentation). In \sref{sc:representation_matrices} we calculate the matrices of the representation of the permutation group $\sigma \mapsto U_\sigma,\, U_\sigma \iota = \iota \cdot \sigma^{-1} $. This allows us to compute the dimensions of the spaces $\InvL{\Hil_j^{\otimes N}}{\lambda}=\InvL{\Hil_j^{\otimes N}}{t^\lambda_{\rm can}}$ (see \sref{sc:dim_irr}). We considered $N$ in the range $2,\ldots,6$ and $j$ in the range $\frac{1}{2},1,\ldots,7$. The results are presented in table \ref{tab:dimensions}. Knowing $\dim \InvL{\Hil_j^{\otimes N}}{\lambda}$, we can predict to some extend the degeneracy of the eigenstates of the volume operator. In \sref{sc:numerical_volume} we discuss the methods we used to calculate the volume operator. In table \ref{tab:volume} we presented the result of our numerical calculation of the eigenvalues of the volume operator. In \sref{sc:deg_eigen} we compare the degeneracy of the eigenstates for the volume matrices, summarized in table \ref{tab:volume}, with the the dimensions of the spaces $\InvL{\Hil_j^{\otimes N}}{\lambda}$, summarized in table \ref{tab:dimensions}. 
\subsection{Tree basis}\label{sc:tree_basis}
A tree basis is built from invariant tensors
\be 
C_{j_1 j_2 j_3} \in \Inv{\Hil_{j_1}\otimes \Hil_{j_2} \otimes \Hil_{j_3}} \quad {\rm and}\quad \epsilon^{j_4} \in \Inv{\Hil^*_{j_4}\otimes \Hil^*_{j_4}}.
\ee
The spaces $\Inv{\Hil_{j_1}\otimes \Hil_{j_2} \otimes \Hil_{j_3}}$ and $\Inv{\Hil^*_{j_4}\otimes \Hil^*_{j_5}}$ are 1-dimensional. Therefore the tensors $C_{j_1 j_2 j_3}$ and $\epsilon^{j_4 j_5}$ are defined uniquely up to scale and phase factors. A common choice is:
\be
C_{j_1 j_2 j_3}^{A_1 A_2 A_3} = \tjsymbol{j_1}{j_2}{j_3}{A_1}{A_2}{A_3},\quad \epsilon^{j}_{A_1 A_2}=(-1)^{j-A_1}\delta_{A_1,-A_2}.
\ee
where $\tjsymbol{j_1}{j_2}{j_3}{A_1}{A_2}{A_3}$ is the 3j-symbol. When the values of the spins $j_1,j_2,j_3,j$ are clear from the context we will suppress them in the notation and write $C^{A_1 A_2 A_3}$ instead of $C_{j_1 j_2 j_3}^{A_1 A_2 A_3}$, $\epsilon_{A_1 A_2}$ instead of $\epsilon^{j}_{A_1 A_2}$.  A tree basis in $\Inv{\Hil_{j_1}\otimes \ldots \otimes \Hil_{j_N}}$ is labelled by a sequence of spins $\vec{k}=(k_1,k_2,\ldots, k_{N-3} )$ satisfying:
\be 
|j_{I} - k_{I-2}|\leq k_{I-1}\leq j_{I} + k_{I-2},\quad k_{I-1}+k_{I-2}+j_I\in \mathbb{N},
\ee
where $I\in \{2,\ldots, N-1\}$, $k_{0}=j_1$, $k_{N-2}=j_N$. The basis element corresponding to the sequence $\vec{k}$ is
\be
\fl \tilde{\iota}_{\vec{k}}^{A_1 \ldots A_N}=C^{A_1 A_2 B_1} \epsilon_{B_1 B'_1} C^{B'_1 A_3 B_2} \epsilon_{B_2 B'_2}\ldots\epsilon_{B_{I-2} B'_{I-2}} C^{B'_{I-2} A_I B_{I-1}} \epsilon_{B_{I-1} B'_{I-1}}\ldots C^{B_{N-3} A_{N-1} A_N},
\ee
where the indices $A_1,\ldots, A_N$ correspond to Hilbert spaces $\Hil_{j_1},\ldots, \Hil_{j_N}$ and the indices $B_1,\ldots, B_{N-3}$ correspond to the Hilbert spaces $\Hil_{k_1},\ldots ,\Hil_{k_{N-3}}$. The basis $\tilde{\iota}_{\vec{k}}$ is orthogonal but not orthonormal. We will use an orthonormal basis
\be
\iota_{\vec{k}} = \prod_{I=1}^{N-2} \sqrt{2k_I +1}\ \tilde{\iota}_{\vec{k}}.
\ee
We put the basis in the co-lexicographic order, i.e. we will say that $\vec{k}< \vec{k'}$ if $k_I<k'_I$ for the last $I$ where $k_I$ and $k'_I$ differ. For example:
\be
(0,\frac{1}{2},0)< (1,\frac{1}{2},0)< (0,\frac{1}{2},1) < (1,\frac{1}{2},1) < (1,\frac{3}{2},1)
\ee
is the ordering of basis in the space $\Inv{\Hil_{\frac{1}{2}}^{\otimes 6}}$. As a side remark, let us notice that the highest element spans the space $\InvLambda{\Hil_{\frac{1}{2}}^{\otimes N}}$, where $T_\lambda=t^\lambda_{\rm can}$.
\subsection{Matrix elements of the representation of the permutation group}\label{sc:representation_matrices}
Having fixed the basis, we can calculate the matrix elements of the representation of the permutation group 
\be
S_N\ni\sigma \mapsto U_\sigma \in \GL{\Inv{\Hil_j^{\otimes N}}},\quad U_\sigma \iota = \iota \cdot \sigma^{-1}.
\ee

Let us focus first on the case of 4-valent intertwiners, i.e. $N=4$. From equation (2.41) from \cite{MakinenRecoupling} it follows that
\be
U_{(23)}\tilde{\iota}_l = \sum_k (2k+1) (-1)^{2j+k+l} \sjsymbol{j}{j}{k}{j}{j}{l} \tilde{\iota}_k,
\ee
where $\sjsymbol{j_1}{j_2}{j_3}{j_4}{j_5}{j_6}$ is the 6j-symbol. Using this formula we can calculate the matrix elements of $U_{(23)}$ in the orthonormal basis $\iota_k$:
\be
U_{(23)}\iota_l =\sum_k {U_{(23)}}\indices{^k_l}\, \iota_k,\quad {U_{(23)}}\indices{^k_l}=\sqrt{(2k+1)(2l+1)} (-1)^{2j+k+l} \sjsymbol{j}{j}{k}{j}{j}{l}.
\ee
The matrix elements $U_{(12)}$ and $U_{(34)}$ can be calculated using the property of the $3j$-symbol that odd permutations of its columns produce a phase factor $(-1)^{j_1+ j_2+ j_3}$:
\be
\tjsymbol{j_1}{j_2}{j_3}{A_1}{A_2}{A_3} =(-1)^{j_1 +j_2 +j_3}\tjsymbol{j_2}{j_1}{j_3}{A_2}{A_1}{A_3}.
\ee
As a result:
\be
{U_{(12)}}\indices{^k_l}={U_{(34)}}\indices{^k_l} = (-1)^{2j+ k} \delta\indices{^k_l},
\ee
where $\delta\indices{^k_l}$ is the Kronecker delta. Since any permutation is a product of the adjacent transpositions, we can calculate all matrices $U_\sigma,$ where $\sigma\in S_4$.

The calculation in the case of 4-valent intertwiners generalizes straightforwardly to $N$-valent intertwiners. The matrix elements of matrices corresponding to adjacent transpositions $U_{(I\, I+1)}$ are the following.
\begin{enumerate}
\item For $I=1$,
\be\label{eq:adj1}
{U_{(1\,2)}}\indices{^{\vec{k}}_{\vec{l}}} = (-1)^{2j+k_1}\prod_{i=1}^{N-3}\delta\indices{^{k_i}_{l_{i}}}.
\ee
\item For any $I\in \{2,3,\ldots, N-2 \}$ the matrices $U_{(I\,I+1)}$ are expressed in terms of the $6j$-symbol:
\be\label{eq:adj2}
\fl {U_{(I\,I+1)}}\indices{^{\vec{k}}_{\vec{l}}} =\sqrt{(2 k_{I-1}+1)(2 l_{I-1}+1)} (-1)^{2 j + k_{I-1} + l_{I-1}} \sjsymbol{k_{I-2}}{j}{k_{I-1}}{k_I}{j}{l_{I-1}} \prod_{i\neq I-1} \delta\indices{^{k_i}_{l_{i}}},
\ee
where we use the convention that $k_0=j_1, k_{N-2}=j_N$. 
\item For $I=N-1$,
\be\label{eq:adj3}
{U_{(N-1\,N)}}\indices{^{\vec{k}}_{\vec{l}}} = (-1)^{2j+k_{N-3}}\prod_{i=1}^{N-3}\delta\indices{^{k_i}_{l_{i}}}.
\ee
\end{enumerate} 
Again, since any permutation can be expressed as a product of adjacent transpositions, any matrix $U_\sigma, \sigma\in S_N$ is a product of the matrices given above.
\subsection{Dimensions of the irreducible spaces}\label{sc:dim_irr}
According to lemma \ref{lm:lemma_decomposition} (\ref{lm:U_splitting}) the Hilbert space $\Inv{\Hil_j^{\otimes N}}$ has the following decomposition:
\be\label{eq:decomposition_Inv}
\Inv{\Hil_j^{\otimes N}}\cong\bigoplus_{\lambda} \InvL{\Hil_j^{\otimes N}}{\lambda}^{\oplus m_\lambda},
\ee
where $\lambda$ runs through all Young diagrams with $N$ boxes and 
\be
\InvL{\Hil_j^{\otimes N}}{\lambda}=\InvL{\Hil_j^{\otimes N}}{t^\lambda_{\rm can}}.
\ee
The number of times $\InvL{\Hil_j^{\otimes N}}{\lambda}$ appears in the decomposition is denoted by $m_\lambda$. It is equal to the dimension of the irreducible representation of $S_N$ corresponding to $\lambda$ and can be calculated using the hook length formula \eqref{eq:dimVlambda}. The dimensions of the spaces $\InvL{\Hil_j^{\otimes N}}{\lambda}$ are equal to the trace of the idempotent $e_\lambda$ obtained by rescaling the Young symmetrizer $c_\lambda$ (see theorem \ref{thm:HilbertSpaceDecomposition}). Specifically, to the groups $P_\lambda$ and $Q_\lambda$ from section \ref{sc:permutation_group_representations} correspond two operators in $\Inv{\Hil_j^{\otimes N}}$:
\be
A_\lambda = \sum_{\sigma \in P_\lambda} U_\sigma, \quad B_\lambda = \sum_{\sigma \in Q_\lambda} \sgn(\sigma) U_\sigma.
\ee
The space $\InvL{\Hil_j^{\otimes N}}{\lambda}$ is the image of the operator $\mathbb{P}_\lambda$ acting in $\Inv{\Hil_j^{\otimes N}}$:
\be
\mathbb{P}_\lambda=\frac{1}{\prod \textrm{(Hook lengths)}}\, B_\lambda\, A_\lambda.
\ee
As a result the dimension of the space $\InvL{\Hil_j^{\otimes N}}{\lambda}$ can be expressed in terms of the trace of $\mathbb{P}_\lambda$:
\be
\dim \InvL{\Hil_j^{\otimes N}}{\lambda} = \frac{1}{\prod \textrm{(Hook lengths)}} \Tr{B_\lambda\, A_\lambda}.
\ee

We calculated the dimensions of the spaces $\dim \InvL{\Hil_j^{\otimes N}}{\lambda} $ for $N=2,3,4,5,6$, $j=\frac{1}{2},\ldots, 7$ together with their multiplicities $m_\lambda$. The matrices $U_\sigma$ are calculated by multiplying the adjacent transposition matrices \eqref{eq:adj1}, \eqref{eq:adj2}, \eqref{eq:adj3}. We calculate them in double accuracy using the library WIGXJPF for the $6j$-symbol \cite{SixJLibrary}. From equations \eqref{eq:adj1}, \eqref{eq:adj2}, \eqref{eq:adj3} it is clear that the matrices are sparse.  We use the Intel\textsuperscript \textregistered  MKL 2019.0 library to multiply and add the sparse matrices. The results are gathered in table \ref{tab:dimensions}. In the table we included only the Young diagrams $\lambda$ such that $\dim\InvL{\Hil_j^{\otimes N}}{\lambda}\neq 0$ for at least one pair $(N,j),\,N\in\{2,3,4,5,6\},\,j\in\{\frac{1}{2},1,\ldots,7\}$. In our numerical calculations the dimensions are double accuracy numbers but are only $\epsilon\approx 10^{-16}$ away from integers. They are rounded to an integer. We will discuss the result in the next subsection. 
\begin{table}[hbt!]
\begin{tabular}{|c|c|c||*{14}{c|}}
\hline
$N$ & $\lambda$ &\backslashbox{$m_\lambda$}{$j$}&$\frac{1}{2}$ & $1$ & $\frac{3}{2}$ & $2$ & $\frac{5}{2}$ & $3$ & $\frac{7}{2}$ & $4$ & $\frac{9}{2}$ & $5$ & $\frac{11}{2}$ & $6$ & $\frac{13}{2}$ & $7$ \\ 
\hline
\multirow{2}{*}{2}&(2) & 1 & 0 & 1 & 0 & 1 & 0 & 1 & 0 & 1 & 0 & 1 & 0 & 1 & 0 & 1 \\  
\cline{2-17}
&(1,1) & 1 & 1 & 0 & 1 & 0 & 1 & 0 & 1 & 0 & 1 & 0 & 1 & 0 & 1 & 0 \\  
\hline
\multirow{2}{*}{3}&(3) & 1 & 0 & 0 & 0 & 1 & 0 & 0 & 0 & 1 & 0 & 0 & 0 & 1 & 0 & 0 \\  
\cline{2-17}
&(1,1,1) & 1 & 0 & 1 & 0 & 0 & 0 & 1 & 0 & 0 & 0 & 1 & 0 & 0 & 0 & 1 \\  
\hline
\multirow{3}{*}{4}&(4) & 1 & 0 & 1 & 1 & 1 & 1 & 2 & 1 & 2 & 2 & 2 & 2 & 3 & 2 & 3 \\  
\cline{2-17}
&(2,2) & 2 & 1 & 1 & 1 & 2 & 2 & 2 & 3 & 3 & 3 & 4 & 4 & 4 & 5 & 5 \\  
\cline{2-17}
&(1,1,1,1) & 1 & 0 & 0 & 1 & 0 & 1 & 1 & 1 & 1 & 2 & 1 & 2 & 2 & 2 & 2 \\  
\hline
\multirow{7}{*}{5}&(5) & 1 & 0 & 0 & 0 & 1 & 0 & 0 & 0 & 2 & 0 & 0 & 0 & 3 & 0 & 0 \\  
\cline{2-17}
&(4,1) & 4 & 0 & 0 & 0 & 1 & 0 & 1 & 0 & 2 & 0 & 3 & 0 & 4 & 0 & 5 \\  
\cline{2-17}
&(3,2) & 5 & 0 & 0 & 0 & 1 & 0 & 1 & 0 & 3 & 0 & 2 & 0 & 6 & 0 & 5 \\  
\cline{2-17}
&(3,1,1) & 6 & 0 & 1 & 0 & 0 & 0 & 3 & 0 & 1 & 0 & 6 & 0 & 3 & 0 & 10 \\  
\cline{2-17}
&(2,2,1) & 5 & 0 & 0 & 0 & 1 & 0 & 0 & 0 & 3 & 0 & 2 & 0 & 5 & 0 & 4 \\  
\cline{2-17}
&(2,1,1,1) & 4 & 0 & 0 & 0 & 0 & 0 & 1 & 0 & 1 & 0 & 2 & 0 & 3 & 0 & 4 \\  
\cline{2-17}
&(1,1,1,1,1) & 1 & 0 & 0 & 0 & 1 & 0 & 0 & 0 & 1 & 0 & 0 & 0 & 2 & 0 & 0 \\  
\hline
\multirow{11}{*}{6}&(6) & 1 & 0 & 1 & 0 & 2 & 0 & 3 & 0 & 4 & 0 & 6 & 0 & 8 & 0 & 10 \\  
\cline{2-17}
&(5,1) & 5 & 0 & 0 & 1 & 0 & 2 & 1 & 4 & 2 & 7 & 4 & 11 & 7 & 16 & 11 \\  
\cline{2-17}
&(4,2) & 9 & 0 & 1 & 0 & 3 & 0 & 6 & 1 & 11 & 3 & 17 & 6 & 26 & 11 & 37 \\  
\cline{2-17}
&(4,1,1) & 10 & 0 & 0 & 1 & 0 & 3 & 1 & 6 & 3 & 11 & 6 & 18 & 11 & 27 & 18 \\  
\cline{2-17}
&(3,3) & 5 & 1 & 0 & 2 & 0 & 4 & 0 & 7 & 0 & 11 & 1 & 16 & 2 & 23 & 4 \\  
\cline{2-17}
&(3,2,1) & 16 & 0 & 0 & 0 & 1 & 2 & 3 & 5 & 7 & 10 & 14 & 18 & 23 & 29 & 36 \\  
\cline{2-17}
&(3,1,1,1) & 10 & 0 & 0 & 0 & 1 & 0 & 3 & 1 & 6 & 3 & 11 & 6 & 18 & 11 & 27 \\  
\cline{2-17}
&(2,2,2) & 5 & 0 & 1 & 0 & 2 & 0 & 4 & 0 & 7 & 0 & 11 & 1 & 16 & 2 & 23 \\  
\cline{2-17}
&(2,2,1,1) & 9 & 0 & 0 & 1 & 0 & 2 & 0 & 5 & 1 & 9 & 2 & 15 & 5 & 23 & 9 \\  
\cline{2-17}
&(2,1,1,1,1) & 5 & 0 & 0 & 0 & 0 & 0 & 1 & 0 & 2 & 1 & 4 & 2 & 7 & 4 & 11 \\  
\cline{2-17}
&(1,1,1,1,1,1) & 1 & 0 & 0 & 0 & 0 & 1 & 0 & 1 & 0 & 2 & 0 & 3 & 0 & 4 & 0 \\  
\hline
\end{tabular}
\caption{We study the decomposition \eqref{eq:decomposition_Inv} $\Inv{\Hil_j^{\otimes N}}\cong\bigoplus_{\lambda} \InvL{\Hil_j^{\otimes N}}{\lambda}^{\oplus m_\lambda}$ for different values of $N=2,3,4,5,6$ and $j=\frac{1}{2},1,\ldots,7$. The table contains the dimensions of the spaces $\InvL{\Hil_j^{\otimes N}}{\lambda}$ and their multiplicities $m_\lambda$. We included only the Young diagrams $\lambda$ such that $\dim\InvL{\Hil_j^{\otimes N}}{\lambda}\neq 0$ for at least one pair $(N,j),\,N\in\{2,3,4,5,6\},\,j\in\{\frac{1}{2},1,\ldots,7\}$.}
\label{tab:dimensions}
\end{table}
\subsection{Volume operator}\label{sc:numerical_volume}
We calculated numerically the matrix of the volume operator in the tree basis and diagonalized it. In order to minimize the risk of making a mistake when entering the formulas or writing a computer code, we used two different methods.

First method is to construct the tree basis from the $3j$-symbols and calculate the matrix elements of the volume operator using the formula for matrix elements of the angular momentum operators in the spin $j$ representation (see for example \cite{Edmonds}). Let us recall that the angular momentum operator $J^3$ is diagonal:
\be
J^3\ket{j A} = A \ket{j\, A}. 
\ee
The operators $J^1$ and $J^2$ are conveniently expressed in term of operators $J_+$ and $J_-$:
\begin{eqnarray}
J^1 = \frac{1}{2} J_+ + \frac{1}{2} J_-,\quad J^2=\frac{1}{2\iu} J_+ - \frac{1}{2\iu} J_-,\\
J_+ \ket{j\, A} = \sqrt{(j-A)(j+A+1)} \ket{j\, A+1},\\ J_- \ket{j\, A} = \sqrt{(j+A)(j-A+1)} \ket{j\, A-1}.
\end{eqnarray}
The operators $J^i_I$ used in the definition of the volume operator can be written in terms of the angular momentum operators:
\be
J^i_I=-\iu\, \id\otimes\ldots\otimes \id\otimes J^i \otimes \id \otimes \ldots \otimes \id.
\ee

Second method is based on two observations, which we made already in section \ref{sc:TrV} when we calculated the overall factor of the volume operator in the case of spin $\frac{1}{2}$ monochromatic intertwiners. We will repeat the formulas for completeness.
\begin{enumerate}
\item  We recall that (see \eqref{eq:UqU}):
\be
q_{\sigma(I) \sigma(J) \sigma(K)} = U_\sigma q_{IJK} U_{\sigma}^{-1}.
\ee
Using such similarity transformations we can express each $q_{IJK}$ in terms of $q_{123}$. 
\item Let us also recall that $q_{123}$ commutes with the isomorphism (see \eqref{eq:Eiso} and \eqref{eq:qEEq}):
\be
\mathcal{E}:\bigoplus_{k}\left(\Inv{\Hil_{j}^{\otimes 3} \otimes \Hil_k} \otimes \Inv{\Hil_k\otimes \Hil_{j}^{\otimes (N-3)}} \right)\to \Inv{\Hil_{j}^{\otimes N}},
\ee
defined by:
\be
(\mathcal{E}(v\otimes w))^{A_1 \ldots A_N}=v^{A_1 A_2 A_3 B_1} \epsilon_{B_1 B_2} w^{B_2 A_4 \ldots A_N}.
\ee
As a result, in order to calculate $q_{123}$ it is enough to know its action in the space of 4-valent intertwiners $\Inv{\Hil_{j}^{\otimes 3} \otimes \Hil_k}$. The second problem has an analytic solution \cite{BrunnemannI,BrunnemannIII,BrunnemannIV}.
\end{enumerate}

We calculated the volume operator for valence $N=4,5,6$ spins $j=\frac{1}{2},1,\frac{3}{2},2,\frac{5}{2}$. The results are presented in table \ref{tab:volume}. Both methods were used for spins $j=\frac{1}{2},1,\frac{3}{2}$ and we verified that with accuracy $10^{-6}$ they give the same results (typically they coincide with accuracy $10^{-15}$ but for zero eigenvalues the accuracy was lower). For higher spins the second method was used only, because it is considerably faster (around 800 times in our tests). Each entry of the table contains an eigenvalue of the volume operator and its degeneracy. The eigenvalue provided in the table is in the units of volume $V_0=\kappa_0 \left(\frac{8\pi G \hbar \gamma}{c^3}\right)^{\frac{3}{2}}$. We assumed that two numerical eigenvalues correspond to the same true eigenvalue if the modulus of their difference does not exceed certain number $\epsilon_{\rm eigen}$. We varied the parameter $\epsilon_{\rm}$ in the range $10^{-15},10^{-14},\ldots,10^{-6}$ and obtained that the degeneracy pattern is stable in the range $10^{-14},10^{-13},\ldots,10^{-6}$.  According to analytic study \cite{BrunnemannI}, in the case of 4-valent intertwiners, the eigenvalues come in pairs except for the zero eigenvalue which is non-degenerate. The zero eigenvalue appears only in the spectrum of odd-dimensional matrices. Since the dimension of the space of 4-valent monochromatic spin $j$ intertwiners is equal $2j+1$, the zero eigenvalue appears if $j\in \mathbb{N}$. This is consistent with the numerical results that we obtained. The non-zero eigenvalues seem to have higher accuracy than zero eigenvalues. For example, our analytic result from \sref{sc:spin_half} coincides with the numerical calculation presented in the table with accuracy $10^{-15}$ but the zero eigenvalues predicted by the analytic study in the 4-valent case \cite{BrunnemannI} are recovered with accuracy $10^{-7}$.
\begin{table}[hbt!]
\begin{tabular}{|*{11}{c|}}
\hline
N & \multicolumn{2}{|c|}{j=1/2} & \multicolumn{2}{|c|}{j=1} & \multicolumn{2}{|c|}{j=3/2} & \multicolumn{2}{|c|}{j=2} & \multicolumn{2}{|c|}{j=5/2}\\ 
\hline 
\multirow{3}{*}{4} & 0.465302 & 2 & 0.000000 & 1 & 0.805927 & 2 & 0.000000 & 1 & 1.133249 & 2\\
\cline{2-11}
 &  &  & 0.930605 & 2 & 1.489473 & 2 & 1.403615 & 2 & 2.061107 & 2\\
\cline{2-11}
 &  &  &  &  &  &  & 2.137338 & 2 & 2.865738 & 2\\
\hline
\multirow{5}{*}{5} &  &  & 1.201406 & 6 &  &  & 0.000000 & 1 &  & \\
\cline{2-11}
 &  &  &  &  &  &  & 2.293525 & 5 &  & \\
\cline{2-11}
 &  &  &  &  &  &  & 2.754087 & 5 &  & \\
\cline{2-11}
 &  &  &  &  &  &  & 2.911094 & 4 &  & \\
\cline{2-11}
 &  &  &  &  &  &  & 3.095988 & 1 &  & \\
\hline
\multirow{14}{*}{6} & 0.930605 & 5 & 0.000000 & 1 & 1.810981 & 5 & 0.000000 & 1 & 2.707704 & 5\\
\cline{2-11}
 &  &  & 1.699044 & 9 & 2.527597 & 5 & 2.741033 & 9 & 3.662116 & 5\\
\cline{2-11}
 &  &  & 1.861210 & 5 & 2.783351 & 10 & 3.107533 & 5 & 4.126456 & 10\\
\cline{2-11}
 &  &  &  &  & 2.945708 & 9 & 3.761054 & 9 & 4.387806 & 9\\
\cline{2-11}
 &  &  &  &  & 3.001623 & 5 & 3.809809 & 1 & 4.392001 & 5\\
\cline{2-11}
 &  &  &  &  &  &  & 3.917444 & 16 & 4.871472 & 5\\
\cline{2-11}
 &  &  &  &  &  &  & 3.973291 & 9 & 5.118623 & 5\\
\cline{2-11}
 &  &  &  &  &  &  & 4.180309 & 10 & 5.144352 & 16\\
\cline{2-11}
 &  &  &  &  &  &  & 4.433841 & 5 & 5.344379 & 16\\
\cline{2-11}
 &  &  &  &  &  &  &  &  & 5.380997 & 10\\
\cline{2-11}
 &  &  &  &  &  &  &  &  & 5.533459 & 5\\
\cline{2-11}
 &  &  &  &  &  &  &  &  & 5.658935 & 9\\
\cline{2-11}
 &  &  &  &  &  &  &  &  & 5.734889 & 10\\
\cline{2-11}
 &  &  &  &  &  &  &  &  & 6.075997 & 1\\
\hline
\end{tabular}

\caption{The table contains the result of our numerical computation of the eigenvalues of the volume operator and their degeneracies. We diagonalized the matrices of the volume operator in the spaces $\Inv{\Hil_j^{\otimes N}}$ for $N=4,5,6$ and $j=\frac{1}{2},1,\frac{3}{2},2,\frac{5}{2}$. Each entry is a pair: an eigenvalue in the units of volume $V_0=\kappa_0 \left(\frac{8\pi G \hbar \gamma}{c^3}\right)^{\frac{3}{2}}$ and its degeneracy. The matrix and its eigenvalues are calculated numerically. We consider two numerical eigenvalues to correspond to the same true eigenvalue if their difference is not greater than $\epsilon_{\rm eigen}=10^{-6}$.}
\label{tab:volume}
\end{table}
\subsection{Degeneracy of the eigenvalues of the volume operator}\label{sc:deg_eigen}
In most of the cases that we calculated, our analysis can be used to predict the pattern of the degeneracy of the eigenstates. For example, by looking at the entries for $N=4, j=1$ in table \ref{tab:dimensions} we infer that the volume operator has
\begin{itemize}
\item $1$ eigenvalue with degeneracy at least $1$ corresponding to $\InvL{\Hil_1^{\otimes 4}}{\lambda}$, $\lambda=(4)$,
\item $1$ eigenvalue with degeneracy at least $2$ corresponding to $\InvL{\Hil_1^{\otimes 4}}{\lambda}$, $\lambda=(2,2)$.
\end{itemize}
It may very well turn out that there is only $1$ eigenvalue with degeneracy $3$ but a quick glimpse at table \ref{tab:volume} reveals that there are $2$ eigenvalues with degeneracies precisely $1$ and $2$. In our examples this pattern repeats in most of the cases: the eigenvalues corresponding to different spaces $\InvL{\Hil_j^{\otimes N}}{\lambda}$ are different. Exceptions are $N=4, j=\frac{3}{2}$ and $N=4, \frac{5}{2}$. The numerical results seem to suggest that the eigenvalues corresponding to spaces for $\lambda=(4)$ and $\lambda=(1,1,1,1)$ are the same.

If $\InvL{\Hil_j^{\otimes N}}{\lambda}$ has dimension greater than $1$ then in principle we do not know if the eigenvalues corresponding to $\InvL{\Hil_j^{\otimes N}}{\lambda}$ are the same or different. In our calculation, it turns out that such eigenvalues are different from each other. For example, in the case $N=6, j=\frac{3}{2}$ we could expect that the volume operator has
\begin{itemize}
\item $1$ eigenvalue with degeneracy $5$ corresponding to $\InvL{\Hil_{\frac{3}{2}}^{\otimes 6}}{\lambda}$, $\lambda=(5,1)$,
\item $1$ eigenvalue with degeneracy $10$ corresponding to $\InvL{\Hil_{\frac{3}{2}}^{\otimes 6}}{\lambda}$, $\lambda=(4,1,1)$,
\item $2$ eigenvalues with degeneracy $5$ \textbf{or} $1$ eigenvalue with degeneracy $10$ corresponding to $\InvL{\Hil_{\frac{3}{2}}^{\otimes 6}}{\lambda}$, $\lambda=(3,3)$,
\item $1$ eigenvalue with degeneracy $9$ corresponding to $\InvL{\Hil_{\frac{3}{2}}^{\otimes 6}}{\lambda}$, $\lambda=(2,2,1,1)$.
\end{itemize}
A quick glimpse at table \ref{tab:volume} reveals that to the Hilbert space $\InvL{\Hil_{\frac{3}{2}}^{\otimes 6}}{\lambda}$, $\lambda=(3,3)$ correspond $2$ eigenvalues with degeneracy $5$.

\section{Summary and outlook}
The Rovelli-Smolin-DePietri volume operator can be defined by its action in the spaces of tensor products $\Hil_{j_1}\otimes \ldots \Hil_{j_N}$, where $\Hil_j$ is the spin $j$ representation space of the SU(2) group. We studied the sector of monochromatic spin $j$ intertwiners $\Inv{\Hil_j^{\otimes N}}$. There is a natural right action of the permutation group $S_N$ on the space $\Hil_j^{\otimes N}$:
\be
(v_1 \otimes \ldots \otimes v_N) \cdot \sigma = v_{\sigma(1)}\otimes \ldots \otimes v_{\sigma(N)},\quad \forall \sigma\in S_N.
\ee
It descends to an action on $\Inv{\Hil_j^{\otimes N}}$. The Rovelli-Smolin-DePietri volume operator commutes with this action. We studied consequences of this fact.

We found a decomposition of the space of spin $j$ monochromatic intertwiners  $\Inv{\Hil_j^{\otimes N}}$:
\be
 \Inv{\Hil_j^{\otimes N}} = \bigoplus_\lambda \bigoplus_{T_\lambda} \InvLambda{\Hil_j^{\otimes N}},
\ee
where $\lambda$ runs through all Young diagrams with $N$ boxes and $T_\lambda$ runs through all standard $\lambda$-tableaux (we recall the definitions of Young diagram and standard $\lambda$-tableau in \sref{sc:Young_diagrams}). In a basis adapted to this decomposition the Rovelli-Smolin-DePietri volume operator takes a block diagonal form. For fixed $\lambda$ the blocks corresponding to different $T_\lambda$ are isospectral. In the case of spin $1/2$ intertwiners the space $\Inv{\Hil_{\frac{1}{2}}^{\otimes N}}$ is non-trivial for $N$ even. In this case, only one Young diagram contributes -- it has two rows with $N/2$ boxes each. Each space $\InvLambda{\Hil_{\frac{1}{2}}^{\otimes N}}$ is 1-dimensional. This shows that the Rovelli-Smolin-DePietri volume operator is proportional to identity in the spaces of spin $1/2$ monochromatic intertwiners $\Inv{\Hil_{\frac{1}{2}}^{\otimes N}}$. 

This paper aims at simplifying the diagonalization problem of the Rovelli-Smolin-DePietri volume operator. In the spin $1/2$ case we solved the diagonalization problem analytically. Since the volume operator is proportion to identity in the spaces of spin $1/2$ intertwiners, it sufficed to calculate the overall factor. We obtained that in the space of spin $1/2$ monochromatic intertwiners $\Inv{\Hil_{\frac{1}{2}}^{\otimes N}}$ the volume operator takes the following form:
\be\label{eq:summary_volume_half_spin}
V_{(\frac{1}{2},\ldots,\frac{1}{2})}=\frac{\kappa_0}{8} \left(\frac{8\pi G \hbar \gamma}{c^3}\right)^{\frac{3}{2}} \sqrt{\frac{\sqrt{3}}{3!} (N-2)N(N+2)}\cdot \id.
\ee
For higher spins the simplification is not so drastic but still significant. The method presented in this paper allows to construct a basis adapted to the decomposition, in which the volume operator takes a block diagonal form. In the numerical study this may turn out to be useful, especially for large matrices. However, further work is needed to make this application practical. The calculation of the matrix elements of the Rovelli-Smolin-DePietri matrix becomes computationally intense when the valence of the intertwiner $N$ grows. One reason is that one should sum $\binom{N}{3}$ matrices. Another reason is that the dimension of the matrices grows when $N$ increases. We expect that in practical application one could overcome the second problem by calculating the blocks corresponding to $\InvLambda{\Hil_{\frac{1}{2}}^{\otimes N}}$ instead of trying to calculate the whole matrix at first place. There is a chance that the first problem can be simplified by using the properties of the basis such as symmetry or antisymmetry of the indices.

The spin $1/2$ case is interesting itself. It could be used as a toy model. It could be also a more fundamental assumption in the models studied -- that the quantum space is built from a large number of small grains described by the spin $1/2$ monochromatic intertwiners. The result concerning spin $1/2$ intertwiners finds an immediate application to our approach to homogeneous-isotropic sector of loop quantum gravity \cite{HomogeneousIsotropicLQG}. Consider scalar constraint operator adding and subtracting loops labeled with spin $1/2$ acting between the spaces of spin $1/2$ monochromatic spin networks. In \cite{HomogeneousIsotropicLQG} we introduced an ad hoc cut-off in the number of loops. Adding a loop at a node $\node$ increases the valence $N$ of the space of intertwiners associated to $\node$ by $2$. From \eqref{eq:summary_volume_half_spin} we immediately see that the cut-off in the number of loops translates into a cut-off in the volume.

Let us notice that the valence $N_\node$ of a node $\node$ can be written in the following form: $N_{\node}=V_\node+2 L_\node$, where $V_{\node}$ is the valence of the node excluding loops and $L_{\node}$ is the number of loops at the node $\node$. As a result, the volume operator from \eqref{eq:summary_volume_half_spin} can be expressed by the number operator $\mathcal{N}_\node$ introduced in \cite{AssanioussiPolymerQuantization,AssanioussiGraphCoherentStates}. This observation can be used to study the coherence properties of the new graph coherent states with respect to the volume operator.

The analysis that we performed applies not only to the volume operator but to any element in the commutator algebra 
\be
B=\{\varphi\in\End{U} : \varphi(v\cdot g)=\varphi(v)\cdot g, \forall v\in U, g\in G\},
\ee
where $U=\Inv{\Hil_j^{\otimes N}}, G=S_N$. For example the Lorentzian part of the Hamiltonian operator proposed in \cite{HamiltonianOperator} has the form:
\be\label{eq:lorentzian_int}
\hat{H}_{\node}^{L_{\rm int}} = \sum_{\link,\link'} \varepsilon(\dot{\link},\dot{\link}') \hat{H}_{\node\, \link,\link'}^L,
\ee
where the sum is over all pairs of links incident at the node $\node$, $\varepsilon(\dot{\link},\dot{\link}')=0$ if the vectors tangent to the links $\link,\link'$ at the node $\node$ are linearly dependent and $1$ otherwise. If we had replaced the operator with
\be\label{eq:lorentzian_ext}
\hat{H}_{\node}^{L_{\rm ext}} = \sum_{\link,\link'} |\hat{H}_{\node\, \link,\link'}^L|,
\ee
we would obtain an element of the commutator algebra, to which our techniques could be applied. It should be verified if both expressions \eqref{eq:lorentzian_int} and \eqref{eq:lorentzian_ext} correspond to the same classical object. One could also make a similar change in the definition of the Euclidean part of the Hamiltonian operator. In this case the operator changes the graph and therefore maps one space of intertwiners into another space of intertwiners. We expect that the techniques from this paper can be extended to this case.
\section*{Acknowledgements}
This work was supported by the National Science Centre, Poland grant No. 2018/28/C/ST9/00157. 
\appendix
\section*{References}
\bibliographystyle{unsrt}
\bibliography{LQG}{}
\end{document}